\documentclass[letterpaper, 10 pt, conference]{ieeeconf}  
\IEEEoverridecommandlockouts   
\usepackage[USenglish]{babel}
\usepackage{xcolor,graphics,graphicx} 
\usepackage{epsfig} 
\usepackage{amsmath} 
\usepackage{amssymb}  
\usepackage{amsfonts}
\usepackage{mathtools}  
\usepackage{optidef,booktabs}
\usepackage{cite}
\usepackage{balance}
\usepackage{url}
\usepackage{lipsum} 
\usepackage{hyperref}
\usepackage[utf8]{inputenc} 
\usepackage[font=small,skip=0pt]{subcaption}
\usepackage[font=small,skip=0pt,compatibility=false]{caption}
\usepackage{tabularray}
\usepackage{tikz}
\usetikzlibrary{decorations.markings}

\UseTblrLibrary{amsmath}

\newtheorem{remark}{Remark}

\newtheorem{theorem}{Theorem}
\newtheorem{corollary}{Corollary}
\newtheorem{definition}{Definition}

\title{\LARGE \bf Computing Scaled Relative Graphs of Discrete-Time LTI Systems: A Frequency-Domain Approach}

\author{Guoming Shi$^{1}$ and  Fredy Ruiz$^{1}$%
\thanks{This work was partially supported by the China Scholarship Council.}
\thanks{$^{1}$The authors are with the Dipartimento di Elettronica, Informazione e Bioingegneria, Politecnico di Milano, 20133 Milan, Italy {\tt\small guoming.shi, fredy.ruiz@polimi.it}}%
}

\begin{document}

\maketitle
\thispagestyle{empty}
\pagestyle{empty}

\begin{abstract}
The scaled relative graph (SRG) represents the joint gain and phase properties of an input--output operator in the complex plane. 
This paper characterizes the SRG closure of causal, stable, square discrete-time linear time-invariant (LTI) systems on one-sided square-summable sequences. 
The Beltrami--Klein image of this closure equals the convex hull of the numerical ranges of the transformed frequency-response matrices. 
For real-coefficient single-input single-output systems, this yields the hyperbolic convex hull of the discrete-time Nyquist locus. 
The proof addresses the Hardy-space restriction imposed by one-sided inputs, which excludes persistent sinusoids, by replacing them with explicitly constructed finite-duration sinusoids. 
This construction shows how points in frequency-wise SRGs arise as limits of operator SRG points generated by admissible inputs. 
The characterization enables model-based computation of the closure from a state-space realization using frequency-response evaluations, numerical ranges, and planar convex hulls, without solving linear matrix inequalities.
\end{abstract}


\section{Introduction}

The scaled relative graph (SRG) was introduced by Ryu, Hannah, and Yin~\cite{ryu2022scaled} as a two-dimensional (2D) geometric tool for analyzing operators and the convergence of optimization algorithms; related results appear in~\cite{huang2020tight}. 
Chaffey et al.~\cite{chaffey2021,chaffey2023} applied SRGs to control systems analysis. Their approach unifies gain and phase in the complex plane and connects frequency-domain methods with nonlinear feedback analysis. 
Subsequent work linked SRGs to numerical ranges~\cite{pates2021scaled}, dissipativity~\cite{de2025dissipativity}, and integral quadratic constraints~\cite{baron2026Equivalence}, and extended them to nonlinear multiple-input multiple-output (MIMO) systems and to soft and hard SRGs~\cite{krebbekx2025nonlinear,krebbekx2025hard,chen2026soft}.

Pates~\cite{pates2021scaled} related linear-operator SRGs to numerical ranges through the Beltrami--Klein map. 
For stable continuous-time single-input single-output (SISO) linear time-invariant (LTI) systems, this connects the SRG to the Nyquist diagram via the Beltrami--Klein model of hyperbolic geometry~\cite{chaffey2023,pates2021scaled}. 
Baron-Prada et al.~\cite{baron2025stability} computed the SRG of the frequency-response matrix at each frequency. Nauta and Pates~\cite{nauta2026closedoperators} characterized SRGs of closed operators through the gains of real shifts of the operator. 
Their discrete-time formulation~\cite{nauta2026Discrete-time} provides data-based and robust computational methods using linear matrix inequalities (LMIs). In the stable case considered here, its input domain is the full space of one-sided square-summable sequences; hence, it characterizes the same SRG closure as the present work. 
We derive an explicit description of this set from a given model's frequency response, without repeatedly solving LMIs for gain bounds.

One-sided square-summable discrete-time input sequences map into the Hardy space $\mathcal{H}_2$, a proper closed subspace of $L_2$ on the unit circle; hence, arbitrary $L_2$ boundary functions cannot be used as inputs. 
Persistent sinusoids are not square summable and are therefore inadmissible. 
Points in the SRG of each frequency-response matrix $G(e^{j\theta})$ are approached using increasingly long truncated sinusoids and thus belong to the operator SRG closure. 
However, they need not belong to the operator SRG itself. 
We address this restriction by constructing finitely supported input sequences that realize the required limiting quadratic-form weights.

Our main result, Theorem~\ref{thm:dt_srg_closure}, establishes that the Beltrami--Klein image of the operator SRG closure equals the convex hull of the numerical ranges of the transformed frequency-response matrices. 
For SISO systems, this yields the hyperbolic convex hull of the classical Nyquist locus, so the continuous-time Nyquist interpretation carries over to discrete time. 
This enables model-based computation of the closure directly from a state-space realization through frequency-response evaluations, numerical ranges, and planar convex hulls. 
Two filter models from~\cite{nauta2026Discrete-time} and two MIMO plants illustrate the underlying frequency-wise SRGs.

Section~\ref{Preliminaries} introduces the notation and SRG definitions. 
Section~\ref{SRG Closure} presents the SRG closure characterization and its SISO Nyquist interpretation. 
Sections~\ref{Examples} and~\ref{Conclusions} provide numerical examples and conclusions, respectively.

\section{Preliminaries}\label{Preliminaries}
\subsection{Notation}
Let $\mathbb{F}$ denote either the real field $\mathbb{R}$ or the complex field $\mathbb{C}$, and let $\mathbb{N}_0:=\{0,1,2,\ldots\}$.
Let $j$ denote the imaginary unit, with $j^2=-1$.
For $z\in\mathbb{C}$, its complex conjugate, real part, imaginary part, and modulus are denoted by $\bar z$, $\Re(z)$, $\Im(z)$, and $|z|$, respectively.
For a complex vector or matrix $A$, $A^\top$ and $A^*$ denote its transpose and conjugate transpose, respectively.
The identity matrix of dimension $n$ is denoted by $I_n$, or simply by $I$ when the dimension is clear.
For a square matrix $A$, its spectral radius is denoted by $\rho(A)$.
For a Hermitian matrix $P$, $P\succ 0$ means that $P$ is positive definite, in which case $P^{1/2}$ denotes its Hermitian positive definite square root.
The closed upper half-plane is denoted by $\overline{\mathbb{C}}_{+}:=\{z\in\mathbb{C}:\Im(z)\geq 0\}$.
The closure of a set $\mathcal{S}$ is denoted by $\mathrm{cl}\,\mathcal{S}$, and its convex hull $\mathrm{co}\,\mathcal{S}$ is the smallest convex set containing $\mathcal{S}$.

\subsection{Signal Spaces}
Let $\mathcal{H}$ denote a Hilbert space over $\mathbb{F}$, equipped with an inner product $\langle\cdot,\cdot\rangle:\mathcal{H} \times \mathcal{H} \rightarrow \mathbb{F}$ and the norm $\|x\|:= \sqrt{\langle x, x\rangle }$. The angle between two nonzero elements $x,y \in \mathcal{H}$ is defined as
\begin{align*}
    \angle(x,y):=\arccos{\left( \frac{\Re \langle x, y\rangle}{\|x\|\,\|y\|}\right)}\in [0, \pi].
\end{align*}

For the discrete-time analysis, we consider the Hilbert space of one-sided square-summable sequences
\begin{align*}
    \ell_2^+(\mathbb{N}_0) := \Big\{ u=(u_k)_{k\in\mathbb{N}_0}:\ u_k\in\mathbb{C}^m,\ \sum_{k=0}^{\infty}u_k^*u_k<\infty\Big\},
\end{align*}
with inner product $\langle u,y\rangle:=\sum_{k=0}^{\infty}u_k^*y_k$.
For $u\in\ell_2^+(\mathbb{N}_0)$, its discrete-time Hardy-space boundary value is denoted by
\begin{align*}
    U(e^{j\theta})= \sum_{k=0}^{\infty}u_k e^{-jk\theta},\qquad 0\leq\theta<2\pi,
\end{align*}
where the series is understood in the $L_2$ boundary-value sense and is defined almost everywhere on the unit circle.
By Parseval's identity, $\langle u,y\rangle=\frac{1}{2\pi}\int_0^{2\pi}U(e^{j\theta})^*Y(e^{j\theta})\,d\theta$. Hence, $u\mapsto U$ is an isometry, with respect to the normalized measure $d\theta/(2\pi)$, from $\ell_2^+(\mathbb{N}_0)$ onto the Hardy space $\mathcal{H}_2$, which is a proper closed subspace of $L_2$ on the unit circle. In particular, not every square-integrable function on the unit circle is the boundary value of an admissible input.

\subsection{Scaled Relative Graphs}
For a linear operator $T$ on $\mathcal{H}$, the SRG is defined as the set of complex numbers~\cite{ryu2022scaled,chaffey2023}
\begin{align*}
    \operatorname{SRG}(T)= \left\{\dfrac{\| {Tv}\|}{\| {v}\|} \exp (\pm j \angle \langle v,Tv \rangle ):v\in\mathcal{H}, v \neq 0 \right\}.
\end{align*}
If $Tv = 0$, the angle is taken to be zero. The set is symmetric with respect to the real axis. For $v\neq0$, we call $z_v:=\frac{\|Tv\|}{\|v\|}e^{j\angle \langle v,Tv \rangle }\in\overline{\mathbb{C}}_+$ the upper-half-plane SRG point generated by $v$. For a matrix $M\in\mathbb{C}^{m\times m}$, $\operatorname{SRG}(M)$ denotes the SRG of the linear map $x\mapsto Mx$ on $\mathbb{C}^m$.

\subsection{The Beltrami--Klein Map and Numerical Range}
To characterize SRGs, we use the Beltrami--Klein map introduced in~\cite{pates2021scaled},
\begin{align*}
    f_{bk}(z)=\frac{(\bar z-j)(z-j)}{1+|z|^2}=\frac{|z|^2-1-2j\,\Re(z)}{1+|z|^2},
\end{align*}
which maps each conjugate pair $\{z,\bar z\}$ to the same point in the closed unit disk and maps real numbers to the unit circle. Its inverse on conjugate-symmetric sets is written as
\begin{align*}
    g_{bk}(w)=\left\{\frac{\Im(w) \pm j\sqrt{1-|w|^2}}{\Re(w) - 1}\right\},\qquad |w|\leq 1,\ w\neq 1,
\end{align*}
and $g_+(w)$ denotes the element of $g_{bk}(w)$ that lies in $\overline{\mathbb{C}}_+$. The map $f_{bk}$ sends the hyperbolic geodesics of the upper half-plane, i.e., vertical half-lines and semicircles centered on the real axis, to straight line segments in the unit disk. Accordingly, following~\cite{chaffey2023,pates2021scaled}, we express the hyperbolic convex hull of a bounded set $\mathcal{S}\subseteq\overline{\mathbb{C}}_+$ as
\begin{align*}
    \mathrm{hconv}(\mathcal{S}):=g_+\big(\mathrm{co}\,f_{bk}(\mathcal{S})\big).
\end{align*}
For a square matrix $M$, define its matrix Beltrami--Klein transform by
\begin{align*}
    f_{bk}(M)&=(I+M^*M)^{-\frac{1}{2}}(M^*-jI)(M-jI)\\
    &\quad\times(I+M^*M)^{-\frac{1}{2}},
\end{align*}
which reduces to the scalar map above when $M$ is a scalar. 

Finally, the numerical range of $M\in\mathbb{C}^{m\times m}$ is defined as the set 
\begin{align*}
    \mathcal{W}(M)=\left\{x^*Mx :x\in \mathbb{C}^m,\|x\|=1\right\}.
\end{align*}
$\mathcal{W}(M)$ is compact, and the Toeplitz--Hausdorff theorem establishes its convexity~\cite{psarrakos2002numerical}.

\section{SRG Closure of Discrete-Time LTI Systems}\label{SRG Closure}
We consider a square LTI system with $m$ inputs and $m$ outputs, described by
\begin{align}
    x_{k+1} &= Ax_k + Bu_k, \qquad x_0 = 0, \label{discrete-time-state} \\
    y_k &= Cx_k + Du_k,    \label{discrete-time-output}
\end{align}
where $A\in\mathbb R^{n\times n}$, $B\in\mathbb R^{n\times m}$, $C\in\mathbb R^{m\times n}$, and
$D\in\mathbb R^{m\times m}$. We consider the complex extension of this real-coefficient realization, allowing complex-valued inputs, and assume $\rho(A)<1$. The induced causal convolution operator $T_d:\ell_2^+(\mathbb{N}_0)\to\ell_2^+(\mathbb{N}_0)$ is then bounded.
The transfer matrix is
\begin{align*}
    G(z) = C(zI-A)^{-1}B + D,
\end{align*}
and its frequency response is $G_d(\theta):=G(e^{j\theta})$, where $\theta$ is the normalized angular frequency $\omega T_s$ modulo $2\pi$, $\omega$ is the angular frequency, and $T_s$ is the sampling period. Define $\Phi_G(\theta):=f_{bk}(G_d(\theta))$.

\begin{theorem}\label{thm:dt_srg_closure}
    Consider the discrete-time LTI system in~\eqref{discrete-time-state}--\eqref{discrete-time-output}, and assume that $A$ is Schur stable. Let $T_d:\ell^+_2(\mathbb{N}_0) \to \ell^+_2(\mathbb{N}_0)$ be the causal convolution operator induced by $G$. Then, the closure of its SRG satisfies
    \begin{align}
        f_{bk}(\mathrm{cl}\,\operatorname{SRG}(T_d)) = \mathrm{co} \bigcup_{\theta \in [0,2\pi)} \mathcal{W}(\Phi_G(\theta )). \label{thm1-bk}
    \end{align}
    Equivalently,
    \begin{align}
        \mathrm{cl}\,\operatorname{SRG}(T_d) &= g_{bk} \left(\mathrm{co} \bigcup_{\theta \in [0,2\pi)} \mathcal{W}(\Phi_G(\theta))\right). \label{thm1-srg}
    \end{align}
\end{theorem}

\begin{proof}
    Let
    \begin{align*}
        P(\theta) &= I + G_d(\theta)^* G_d(\theta), \\
        Q(\theta) &= G_d(\theta)^* G_d(\theta) - I - j\big(G_d(\theta)^* + G_d(\theta)\big).
    \end{align*}
    Since $P(\theta)\succ 0$ for every $\theta\in[0,2\pi)$, we have
    \begin{align}
        \Phi_G(\theta)=P(\theta)^{-1/2}Q(\theta)P(\theta)^{-1/2}.\label{phi-G-representation}
    \end{align}
    Because $G_d(\theta)$ is continuous and $2\pi$-periodic, both $P(\theta)$ and $\Phi_G(\theta)$ are continuous periodic matrix functions. Hence, the set
    \begin{align*}
        \bigcup_{\theta\in[0,2\pi)}\mathcal{W}(\Phi_G(\theta))
    \end{align*}
    is compact. Indeed, it is the image of the compact set $[0,2\pi]\times\{v\in\mathbb{C}^m:\|v\|=1\}$ under the continuous map
    \begin{align*}
        (\theta,v)\mapsto v^*\Phi_G(\theta)v.
    \end{align*}
    Therefore, its convex hull
    \begin{align}
        K_G:=\mathrm{co}\bigcup_{\theta\in[0,2\pi)}\mathcal{W}(\Phi_G(\theta))\label{KG-definition}
    \end{align}
    is also compact, since it is the convex hull of a compact subset of $\mathbb{C}\simeq\mathbb{R}^2$.

    We first prove
    \begin{align}
        f_{bk}(\operatorname{SRG}(T_d)) \subseteq K_G. \label{forward-inclusion}
    \end{align}
    Let $u\in\ell_2^+(\mathbb{N}_0)$, $u\neq0$, and let $y=T_du$. Denote by $U(e^{j\theta})$ and $Y(e^{j\theta})$ the boundary values of the corresponding Hardy-space transforms. Since $T_d$ is the causal LTI operator induced by $G$, we have
    \begin{align}
        Y(e^{j\theta}) = G(e^{j\theta})U(e^{j\theta}) \label{frequency-relation}
    \end{align}
    for almost every $\theta\in[0,2\pi)$.
    By Parseval's identity,
    \begin{align}
        \|u\|^2 &= \frac{1}{2\pi} \int_{0}^{2\pi} \|U(e^{j\theta})\|^2 \, d\theta, \label{parseval-input} \\
        \|T_d u\|^2 &= \frac{1}{2\pi} \int_{0}^{2\pi} \|G(e^{j\theta})U(e^{j\theta})\|^2 \, d\theta, \label{parseval-output} \\
        \langle u, T_d u \rangle &= \frac{1}{2\pi} \int_{0}^{2\pi} U(e^{j\theta})^* G(e^{j\theta}) U(e^{j\theta}) \, d\theta. \label{parseval-inner}
    \end{align}

    Let $z_u$ be the upper-half-plane SRG point generated by $u$, and define
    \begin{align*}
        r^2 = \frac{\|T_d u\|^2}{\|u\|^2}, \qquad c = \frac{\Re\langle u, T_d u \rangle}{\|u\|^2}.
    \end{align*}
    Then
    \begin{align}
        z_u = c+j\sqrt{r^2-c^2}.
    \end{align}
    Applying the scalar Beltrami--Klein map gives
    \begin{align}
        f_{bk}(z_u)=\frac{r^2-1-2jc}{r^2+1}. \label{bk-srg-point}
    \end{align}
    Substituting~\eqref{parseval-input}--\eqref{parseval-inner} into~\eqref{bk-srg-point} and suppressing the frequency argument for readability yields
    \begin{align}
        f_{bk}(z_u) = \frac{\int_{0}^{2\pi} U^*(G^*G - I - j(G^*+ G))U\, d\theta}
        {\int_{0}^{2\pi} U^*(I + G^*G)U \, d\theta}.\label{f-bk-zu}
    \end{align}
    Hence, using the definitions of $P(\theta)$ and $Q(\theta)$,
    \begin{align}
        f_{bk}(z_u) = \frac{\int_0^{2\pi} U(e^{j\theta})^*Q(\theta) U(e^{j\theta})\,d\theta}{\int_0^{2\pi} U(e^{j\theta})^*P(\theta) U(e^{j\theta})\,d\theta}.   \label{pq-quotient}
    \end{align}

    Let $\xi(\theta) = P(\theta)^{1/2}U(e^{j\theta})$. Using~\eqref{phi-G-representation}, we can rewrite~\eqref{pq-quotient} as
    \begin{align}
        f_{bk}(z_u) = \frac{\int_{0}^{2\pi} \xi(\theta)^*\Phi_G(\theta)\xi(\theta) \, d\theta}{\int_{0}^{2\pi} \|\xi(\theta)\|^2 \, d\theta}. \label{weighted-numerical-range}
    \end{align}
    For each fixed $\theta$ with $\xi(\theta)\neq 0$,
    \begin{align}
        \frac{\xi(\theta)^*\Phi_G(\theta)\xi(\theta)}{\|\xi(\theta)\|^2} \in \mathcal{W}(\Phi_G(\theta)).
    \end{align}
    Therefore,~\eqref{weighted-numerical-range} is a weighted average of points belonging to
    \begin{align}
        \bigcup_{\theta \in [0,2\pi)} \mathcal{W}(\Phi_G(\theta)),
    \end{align}
    with nonnegative weights proportional to
    \begin{align*}
        \|\xi(\theta)\|^2 = U(e^{j\theta})^* P(\theta) U(e^{j\theta}).
    \end{align*}
    Since $K_G$ is closed and convex, it contains every such weighted average. Hence,
    \begin{align}
        f_{bk}(z_u)\in K_G,
    \end{align}
    and since $f_{bk}(\bar z_u)=f_{bk}(z_u)$, this proves~\eqref{forward-inclusion}.

    We now prove the reverse inclusion. Let $\eta \in K_G$. Since $K_G$ is the convex hull in~\eqref{KG-definition}, there exist finitely many frequencies $\theta_1,\ldots,\theta_K\in[0,2\pi)$, unit vectors $v_1,\ldots,v_K\in\mathbb{C}^m$, and weights $\alpha_k\geq0$, with
    \begin{align}
        \sum_{k=1}^K\alpha_k=1, \label{alpha-sum}
    \end{align}
    such that
    \begin{align}
        \eta = \sum_{k=1}^K \alpha_k v_k^*\Phi_G(\theta_k)v_k. \label{eta-convex-combination}
    \end{align}
    Terms with zero weight may be omitted. Moreover, if several terms correspond to the same frequency, they can be combined by the convexity of the numerical range $\mathcal{W}(\Phi_G(\theta))$. We may therefore assume that $\theta_1,\ldots,\theta_K$ are distinct modulo $2\pi$.
    Set
    \begin{align}
        a_k = P(\theta_k)^{-1/2}v_k, \qquad k=1,\ldots, K. \label{ak-definition}
    \end{align}
    For $N\geq1$, define the one-sided, finitely supported input
    \begin{align}
        u_N(t) =
        \begin{cases}
            \displaystyle
            \frac{1}{\sqrt{N}} \sum_{k=1}^K \sqrt{\alpha_k}\, a_k e^{jt\theta_k}, &0\leq t<N,\\
                0,    &t\geq N.
        \end{cases}    \label{uN-definition}
    \end{align}
    By construction, $u_N\in\ell_2^+(\mathbb{N}_0)$. Its Hardy-space boundary value is
    \begin{align}
        U_N(e^{j\theta}) = \sum_{k=1}^K \sqrt{\alpha_k}\, a_k d_N(\theta-\theta_k), \label{UN-definition}
    \end{align}
    where
    \begin{align}
        d_N(\varphi) = \frac{1}{\sqrt{N}} \sum_{t=0}^{N-1} e^{-jt\varphi}.\label{dN-definition}
    \end{align}

    We establish the following localization property, suppressing the dependence of $U_N$ and $M$ on $\theta$ for readability. For every continuous $2\pi$-periodic matrix function $M(\theta)$,
    \begin{align}
        \frac{1}{2\pi}  \int_0^{2\pi} U_N^* M U_N\,d\theta
        \longrightarrow \sum_{k=1}^K \alpha_k a_k^*M(\theta_k)a_k \label{localization-limit}
    \end{align}
    as $N\to\infty$. Substituting~\eqref{UN-definition} produces diagonal terms and cross terms. For the diagonal terms,
    \begin{align*}
        |d_N(\theta-\theta_k)|^2
    \end{align*}
    is the Fej\'er kernel, normalized so that
    \begin{align*}
        \frac{1}{2\pi} \int_0^{2\pi} |d_N(\theta-\theta_k)|^2\,d\theta = 1.
    \end{align*}
    The Fej\'er kernels concentrate their mass at $\theta=\theta_k$, and continuity of $M$ therefore gives
    \begin{align*}
        \frac{1}{2\pi} \int_0^{2\pi} |d_N(\theta-\theta_k)|^2 a_k^*M(\theta)a_k\,d\theta
        \longrightarrow a_k^*M(\theta_k)a_k.
    \end{align*}
    For $k\neq l$, the frequencies $\theta_k$ and $\theta_l$ are separated. The $L_2$-mass of $d_N(\theta-\theta_k)$ concentrates near $\theta_k$, while that of $d_N(\theta-\theta_l)$ concentrates near $\theta_l$. Since $M$ is bounded, the Cauchy--Schwarz inequality implies that the corresponding cross terms converge to zero. 
    This proves~\eqref{localization-limit}.

    Applying~\eqref{localization-limit} first with $M=P$ gives
    \begin{align}
        \frac{1}{2\pi} \int_0^{2\pi} U_N^*P U_N\,d\theta  &\longrightarrow
        \sum_{k=1}^K \alpha_k a_k^*P(\theta_k)a_k \nonumber\\
        &= \sum_{k=1}^K \alpha_k v_k^*v_k = \sum_{k=1}^K\alpha_k = 1. \label{P-limit}
    \end{align}
    In particular, the denominator in~\eqref{pq-quotient} with $U=U_N$ is nonzero for all sufficiently large $N$, so $u_N\neq0$.

    Applying~\eqref{localization-limit} with $M=Q$, and using~\eqref{ak-definition} and~\eqref{phi-G-representation}, gives
    \begin{align}
        \frac{1}{2\pi} \int_0^{2\pi} U_N^*Q U_N\,d\theta &\longrightarrow
        \sum_{k=1}^K \alpha_k a_k^*Q(\theta_k)a_k \nonumber\\
        &= \sum_{k=1}^K \alpha_k v_k^*\Phi_G(\theta_k)v_k = \eta.  \label{Q-limit}
    \end{align}
    Let $z_{u_N}$ denote the upper-half-plane SRG point generated by $u_N$.
    By the same calculation as in~\eqref{pq-quotient},
    \begin{align}
        f_{bk}(z_{u_N}) = \frac{\int_0^{2\pi} U_N^*Q U_N\,d\theta}
        {\int_0^{2\pi} U_N^*P U_N\,d\theta}. \label{f-bk-uN}
    \end{align}
    It follows from~\eqref{P-limit} and~\eqref{Q-limit} that
    \begin{align}
        f_{bk}(z_{u_N}) \longrightarrow \eta. \label{eta-limit}
    \end{align}
    Hence,
    \begin{align}
        K_G \subseteq \mathrm{cl}\, f_{bk}(\operatorname{SRG}(T_d)). \label{reverse-inclusion}
    \end{align}
    Combining~\eqref{forward-inclusion} and~\eqref{reverse-inclusion},
    and using the fact that $K_G$ is closed, gives
    \begin{align}
        \mathrm{cl}\, f_{bk}(\operatorname{SRG}(T_d)) = K_G.\label{closure-image-equality}
    \end{align}
    It remains to relate this equality to $f_{bk}(\mathrm{cl}\,\operatorname{SRG}(T_d))$. Since $T_d$ is bounded, its operator norm satisfies
    \begin{align*}
        \gamma:=\|T_d\|<\infty.
    \end{align*}
    Every $z\in\operatorname{SRG}(T_d)$ satisfies $|z|\leq\gamma$. Therefore, $\mathrm{cl}\,\operatorname{SRG}(T_d)$ is compact. Since $f_{bk}$ is continuous on the complex plane,
    \begin{align}
        f_{bk} \left(\mathrm{cl}\,\operatorname{SRG}(T_d) \right) &=\mathrm{cl}\, f_{bk} \left(\operatorname{SRG}(T_d) \right) \nonumber\\
        &= K_G. \label{final-bk-equality}
    \end{align}
    This proves~\eqref{thm1-bk}.

    Finally, for every $z\in\mathrm{cl}\,\operatorname{SRG}(T_d)$,
    \begin{align}
        \Re\,f_{bk}(z) = \frac{|z|^2-1}{|z|^2+1} \leq \frac{\gamma^2-1}{\gamma^2+1} <1.
        \label{avoid-one}
    \end{align}
    Hence, $K_G$ does not contain the point $1$, at which the inverse Beltrami--Klein parametrization becomes singular.
    Applying $g_{bk}$ to~\eqref{final-bk-equality}, and using the conjugate symmetry of $\operatorname{SRG}(T_d)$, yields~\eqref{thm1-srg}.
\end{proof}

\begin{remark}\label{rem:neg-freq}
    For real-coefficient systems, $G_d(-\theta)=\overline{G_d(\theta)}$, where frequencies are understood modulo $2\pi$. The matrix Beltrami--Klein transform consequently satisfies $\Phi_G(-\theta)=\Phi_G(\theta)^\top$.
    Since $\mathcal{W}(M^\top)=\mathcal{W}(M)$, negative frequencies contribute no additional points to $K_G$.
    Thus, $\theta\in[0,\pi]$ suffices to construct $K_G$, and $g_{bk}(K_G)$ supplies both conjugate halves of the operator SRG closure.
\end{remark}

\begin{remark}\label{rem:freqwise}
    For every $\theta$, the numerical range $\mathcal{W}(\Phi_G(\theta))$ is exactly the Beltrami--Klein image of the SRG of the constant matrix $G_d(\theta)$, i.e., $f_{bk}(\operatorname{SRG}(G_d(\theta)))=\mathcal{W}(\Phi_G(\theta))$~\cite{pates2021scaled}. 
    This follows from the finite-dimensional version of the computation leading to~\eqref{weighted-numerical-range}. 
    Hence, Theorem~\ref{thm:dt_srg_closure} states that the operator SRG closure is obtained by convexifying, in the Beltrami--Klein disk, the frequency-wise SRGs used in~\cite{baron2025stability}. In particular, $\operatorname{SRG}(G_d(\theta))\subseteq\mathrm{cl}\,\operatorname{SRG}(T_d)$ for every $\theta\in[0,2\pi)$.
\end{remark}

\begin{definition}
    For a SISO transfer function $G$ inducing $T_d$, define its discrete-time Nyquist locus by
    \begin{align*}
        \operatorname{Nyquist}_d(T_d):=\{G(e^{j\theta}):\theta\in[0,2\pi)\}.
    \end{align*}
\end{definition}

\begin{corollary}\label{cor:siso}
    Consider the real-coefficient SISO case ($m=1$) of Theorem~\ref{thm:dt_srg_closure}. Then
    \begin{align}
        \mathcal{W}(\Phi_G(\theta)) = \{f_{bk}(G_d(\theta))\}, \qquad \theta\in[0,2\pi).
    \end{align}
    Therefore,
    \begin{align}
        f_{bk}(\mathrm{cl}\,\operatorname{SRG}(T_d)) = \mathrm{co}\{f_{bk}(G_d(\theta)) : \theta \in [0, 2\pi)\}.
    \end{align}
    Writing $N_d:=\operatorname{Nyquist}_d(T_d)$, we obtain
    \begin{align}
        \mathrm{cl}\,\operatorname{SRG}(T_d) \cap \overline{\mathbb{C}}_{+} = \mathrm{hconv}(N_d \cap \overline{\mathbb{C}}_{+}).
    \end{align}
\end{corollary}
\begin{proof}
    The numerical range of a scalar is the singleton containing that scalar, and for $m=1$ the matrix transform coincides with the scalar map. 
    The first two identities therefore follow from Theorem~\ref{thm:dt_srg_closure}. 
    Since $G$ has real coefficients, $N_d$ is conjugate symmetric and $f_{bk}(N_d)=f_{bk}(N_d\cap\overline{\mathbb{C}}_+)$. Since $\mathrm{cl}\,\operatorname{SRG}(T_d)$ is conjugate symmetric and $g_+$ inverts $f_{bk}$ on $\overline{\mathbb{C}}_+$, applying $g_+$ and the definition of $\mathrm{hconv}$ proves the final identity.
\end{proof}

\section{Numerical Examples}\label{Examples}
\subsection{SISO Filters}\label{sec:filters}
Consider two SISO models from~\cite{nauta2026Discrete-time}: a low-pass filter and a high-pass filter. The low-pass filter has the state-space representation
\begin{align}
    x_{k+1} &= \left[\begin{array}{cc} 0.94 & -0.33\\ 1 & 0 \end{array}\right] x_k + \left[\begin{array}{c} 1 \\ 0 \end{array}\right] u_k, \label{lowpass_state} \\
    y_k &= \left[\begin{array}{cc} 0.29 & 0.07 \end{array}\right] x_k + 0.10u_k,    \label{lowpass_out}
\end{align}
and the high-pass filter has the state-space representation
\begin{align}
    x_{k+1} &= \left[\begin{array}{cc} 0.94 & -0.33\\ 1 & 0 \end{array}\right] x_k + \left[\begin{array}{c} 1 \\ 0 \end{array}\right] u_k, \label{highpass_state} \\
    y_k &= \left[\begin{array}{cc} -0.60 & 0.38 \end{array}\right] x_k + 0.57u_k.    \label{highpass_out}
\end{align}
Both filters are Schur stable, and Corollary~\ref{cor:siso} characterizes their SRG closures. Figs.~\ref{fig:highpass_2d} and~\ref{fig:lowpass_2d} show these closures, and Figs.~\ref{fig:highpass_3d} and~\ref{fig:lowpass_3d} show the frequency-wise SRGs over frequency. 
Fig.~\ref{fig:lmi_hplp} shows the LMI-based SRGs of the two filters, obtained with the method of~\cite{nauta2026Discrete-time}.

Although the two filters have opposite frequency responses, their SRG closures, as well as their LMI-based SRGs, are nearly identical.
By Corollary~\ref{cor:siso}, the closure depends only on the Nyquist locus as a set and does not record the frequency at which each point is attained.
The frequency-wise SRGs in Figs.~\ref{fig:highpass_3d} and~\ref{fig:lowpass_3d} retain this information and distinguish the two filters: as $\theta$ increases from $0$ to $\pi$, the SRG of the low-pass filter moves from the point near $1$ toward the origin, whereas that of the high-pass filter moves in the opposite direction.

\begin{figure}[t]
    \centering
    \begin{subfigure}[b]{0.48\columnwidth}
        \centering
        \includegraphics[width=\linewidth]{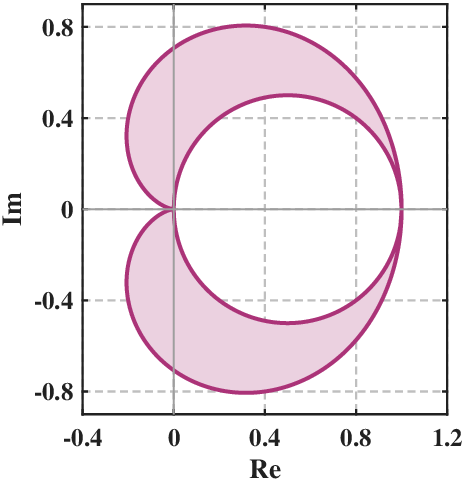}
        \caption{High-pass filter (2D).}
        \label{fig:highpass_2d}
    \end{subfigure}
    \hfill
    \begin{subfigure}[b]{0.48\columnwidth}
        \centering
        \includegraphics[width=\linewidth]{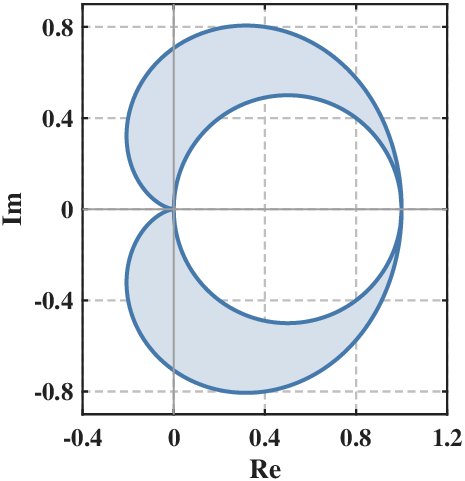}
        \caption{Low-pass filter (2D).}
        \label{fig:lowpass_2d}
    \end{subfigure}
    \par\medskip
    \begin{subfigure}[b]{0.48\columnwidth}
        \centering
        \includegraphics[width=\linewidth]{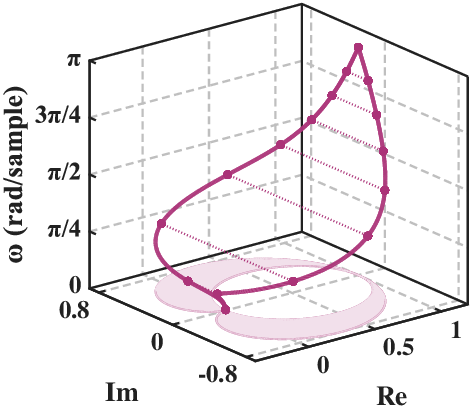}
        \caption{High-pass filter (3D).}
        \label{fig:highpass_3d}
    \end{subfigure}
    \hfill
    \begin{subfigure}[b]{0.48\columnwidth}
        \centering
        \includegraphics[width=\linewidth]{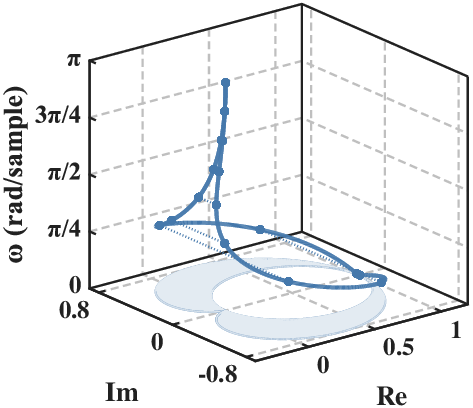}
        \caption{Low-pass filter (3D).}
        \label{fig:lowpass_3d}
    \end{subfigure}
    \caption{Example A. SRG closures (2D) and frequency-wise SRGs (3D) of the high-pass and low-pass filters.}
    \label{fig:filters}
\end{figure}

\begin{figure}[t]
    \centering
    \includegraphics[width=0.95\columnwidth]{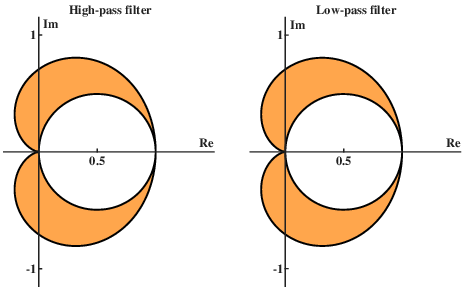}
    \caption{Example A. LMI-based SRGs of the high-pass and low-pass filters.}
    \label{fig:lmi_hplp}
\end{figure}

\subsection{\texorpdfstring{$2\times 2$}{2x2} MIMO Plant}\label{sec:mimo}
Consider the two-input two-output plant from~\cite{cordoba2025mimo},
\begin{align}
    P(z) = \frac{1}{z-0.9048}\left[\begin{array}{cc} 0.09516 & 0.03807\\ -0.02974 & 0.04758 \end{array}\right]. \label{mimo_plant}
\end{align}
A state-space realization~\eqref{discrete-time-state}--\eqref{discrete-time-output} is given by $A=0.9048\,I_2$, $B=I_2$, $D=0$, and $C$ equal to the numerator matrix in~\eqref{mimo_plant}. 
The plant is Schur stable. Fig.~\ref{fig:mimo_freq_2d} shows the union of the frequency-wise SRGs $\operatorname{SRG}(P(e^{j\theta}))$, and Fig.~\ref{fig:mimo_freq_3d} shows their distribution over frequency. 

Unlike in the SISO case, each frequency-wise SRG is now a region.
For this plant, the Beltrami--Klein image of their union is already convex, so by Theorem~\ref{thm:dt_srg_closure} the union equals $\mathrm{cl}\,\operatorname{SRG}(T_d)$; accordingly, it coincides with the LMI-based SRG in Fig.~\ref{fig:mimo_lmi_2d}, obtained with the method of~\cite{nauta2026Discrete-time}.

\begin{figure}[t]
    \centering
    \begin{subfigure}[b]{0.48\columnwidth}
        \centering
        \includegraphics[width=\linewidth]{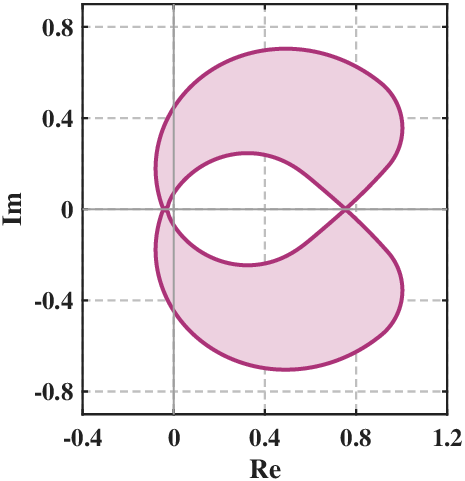}
        \caption{Frequency-wise SRGs (2D).}
        \label{fig:mimo_freq_2d}
    \end{subfigure}
    \hfill
    \begin{subfigure}[b]{0.48\columnwidth}
        \centering
        \includegraphics[width=\linewidth]{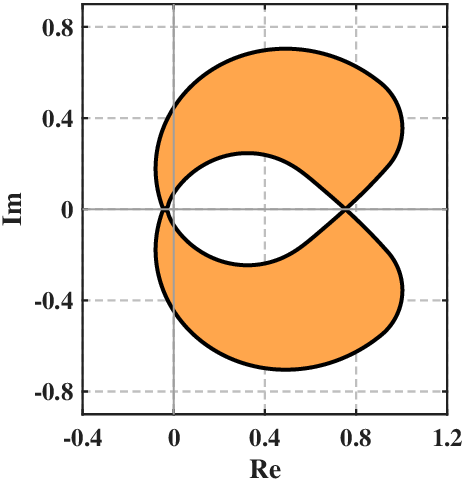}
        \caption{LMI-based SRG.}
        \label{fig:mimo_lmi_2d}
    \end{subfigure}
    \par\medskip
    \begin{subfigure}[b]{0.6\columnwidth}
        \centering
        \includegraphics[width=\linewidth]{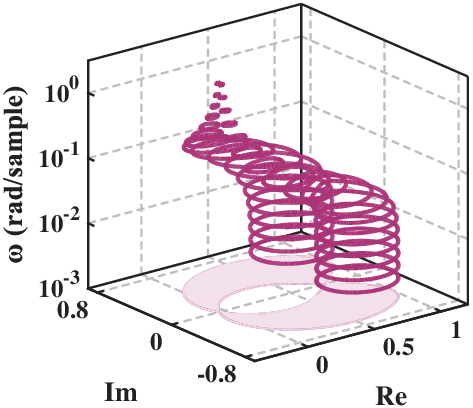}
        \caption{Frequency-wise SRGs (3D).}
        \label{fig:mimo_freq_3d}
    \end{subfigure}
    \caption{Example B. SRGs of the MIMO plant~\eqref{mimo_plant}.}
    \label{fig:mimo}
\end{figure}

 \subsection{\texorpdfstring{$4\times 4$}{4x4} MIMO Plant}
Consider next the four-input four-output plant from~\cite{baron2025mixed},
\begin{equation}
    \setlength{\arraycolsep}{1.5pt}\setlength{\nulldelimiterspace}{0pt}
    H(s) = \begin{bmatrix}
        \frac{1}{(s+1)^3(s+2)} & \frac{2}{s+3} & \frac{4}{(s+1)(s+4)} & \frac{1}{(s+1)^2(s+2)^2}\\[3pt]
        \frac{2}{s+5} & \frac{3}{(s+3)(s+4)} & \frac{3}{(s+1)^2(s+2)} & \frac{3}{s+4}\\[3pt]
        \frac{1}{(s+1)^3} & \frac{3}{s+5} & \frac{1}{(s+1)(s+3)} & \frac{2}{(s+3)(s+4)}\\[3pt]
        \frac{1}{s+5} & \frac{2}{(s+1)^5(s+2)} & \frac{1}{(s+1)(s+2)} & \frac{1}{s+1}
    \end{bmatrix}\!, \label{mimo4_plant}
\end{equation}
discretized by a zero-order hold (ZOH) with sampling period $T_s=0.1$~s. 
A balanced minimal realization~\eqref{discrete-time-state}--\eqref{discrete-time-output} of the discretized plant has $23$ states and $D=0$. 
Since $\rho(A)=e^{-0.1}\approx0.9048<1$, Theorem~\ref{thm:dt_srg_closure} applies. 
Fig.~\ref{fig:mimo4} shows the resulting SRGs. 

The union of the frequency-wise SRGs and the LMI-based SRG share the same outer boundary, but near the point $0.6$ on the real axis the union leaves a wedge-shaped region uncovered, part of which the LMI-based SRG covers.
Theorem~\ref{thm:dt_srg_closure} explains this difference: a frequency-wise SRG accounts only for inputs concentrated at a single frequency, whereas $\mathrm{cl}\,\operatorname{SRG}(T_d)$ also contains the points generated by inputs~\eqref{uN-definition} that combine several frequencies, i.e., convex combinations of numerical ranges at different frequencies in the Beltrami--Klein disk.

\begin{figure}[t]
    \centering
    \begin{subfigure}[b]{0.48\columnwidth}
        \centering
        \includegraphics[width=\linewidth]{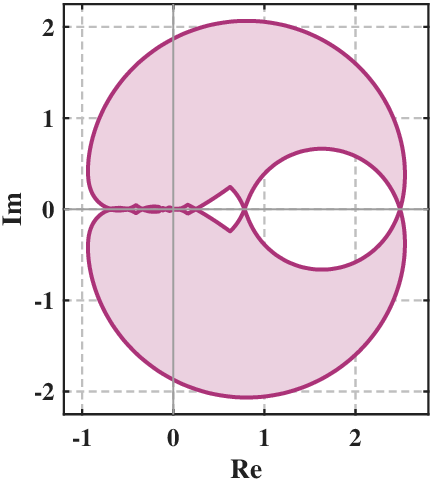}
        \caption{Frequency-wise SRGs (2D).}
        \label{fig:mimo4_freq_2d}
    \end{subfigure}
    \hfill
    \begin{subfigure}[b]{0.48\columnwidth}
        \centering
        \includegraphics[width=\linewidth]{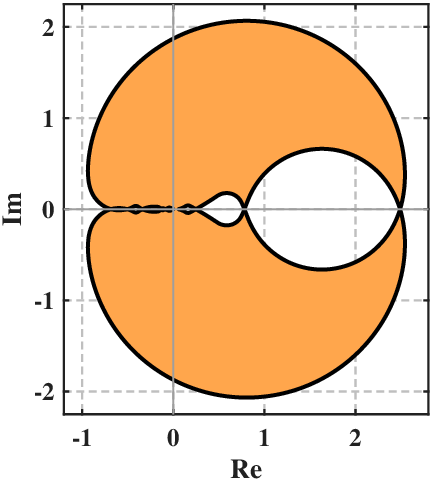}
        \caption{LMI-based SRG.}
        \label{fig:mimo4_lmi_2d}
    \end{subfigure}
    \par\medskip
    \begin{subfigure}[b]{0.6\columnwidth}
        \centering
        \includegraphics[width=\linewidth]{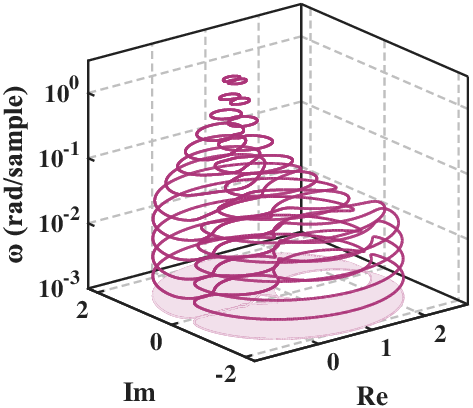}
        \caption{Frequency-wise SRGs (3D).}
        \label{fig:mimo4_freq_3d}
    \end{subfigure}
    \caption{Example C. SRGs of the MIMO plant~\eqref{mimo4_plant} after ZOH discretization with $T_s=0.1$~s.}
    \label{fig:mimo4}
\end{figure}

Table~\ref{tab:time} compares the computation times of the two methods.
The frequency-domain method uses $1000$ frequency points, and the LMI-based method~\cite{nauta2026Discrete-time} uses $1000$ points on the real axis, with MOSEK as the solver.

\begin{table}[t]
    \centering
    \caption{Computation time (s)}
    \label{tab:time}
    \begin{tabular}{lcc}
        \toprule
        Example & Frequency-domain & LMI-based \\
        \midrule
        A (high-pass) & 0.611 & 5.383 \\
        A (low-pass)  & 0.293 & 4.996 \\
        B  & 1.359 & 3.386 \\
        C  & 1.747 & 401.185 \\
        \bottomrule
    \end{tabular}
\end{table}

\section{Conclusions}\label{Conclusions}
This paper characterizes the SRG closure of stable, square discrete-time LTI systems as the Beltrami--Klein convexification of the frequency-wise SRGs, which reduces to the hyperbolic convex hull of the Nyquist locus in the SISO case.
When a model is available, this frequency-domain computation is preferable to the LMI-based approach, since it requires only frequency-response evaluations, numerical ranges, and planar convex hulls.
Future work will address identifying SRGs from frequency-response data.

\bibliographystyle{IEEEtran}
\bibliography{main_ICSC}

\end{document}